\documentclass[12pt]{amsart}

\usepackage{times,fullpage}
\usepackage{color} 
\usepackage[pdftex]{graphicx}
\usepackage[pdftex]{hyperref}
\usepackage{makecell}
\hypersetup{
    unicode=false,          
    pdftoolbar=true,        
    pdfmenubar=true,        
    pdffitwindow=false,     
    pdfstartview={FitH},    
    pdfnewwindow=true,      
    colorlinks=true,        
    linkcolor=red,          
    citecolor=blue,         
    filecolor=magenta,      
    urlcolor=cyan           
}
\usepackage{amsmath}
\usepackage{amsthm}
\usepackage{amssymb}
\usepackage{amsbsy}
\usepackage{amscd}
\usepackage{enumerate}
\usepackage{verbatim}

\theoremstyle{plain}
\newtheorem{thm}{Theorem}[section]

\theoremstyle{definition}
\newtheorem{defn}[thm]{Definition}

\newtheorem{conj}[thm]{Conjecture}

\newcommand{\tr}{\mathrm{tr}}

\newcommand{\calV}{\mathcal{V}}

\newcounter{mnotecount}[section]
\let\oldmarginpar\marginpar
\renewcommand\marginpar[1]{\-\oldmarginpar[\raggedleft\footnotesize #1]%
{\raggedright\footnotesize #1}}

\begin{document}

\title[When do Spacetimes Have CMC Slices]
      {When do Spacetimes Have Constant Mean Curvature Slices?}

\author[J. Dilts]{James Dilts}
\email{jdilts@ucsd.edu}

\author[M. Holst]{Michael Holst}
\email{mholst@ucsd.edu}

\address{Department of Mathematics\\
         University of California, San Diego\\
         La Jolla CA 92093}

\thanks{JD was supported in part by NSF DMS/RTG Award 1345013 and DMS/FRG Award 1262982.}
\thanks{MH was supported in part by NSF DMS/FRG Award 1262982 and NSF DMS/CM Award 1620366.}

\date{\today}
\subjclass[2010]{00000, 00000}

\begin{abstract} 
Many results in mathematical relativity, including results for both the initial data problem and for the evolution problem, rely on the existence of a constant mean curvature (CMC) Cauchy surface in the underlying spacetime. 
However, it is known that some spacetimes have no CMC Cauchy surfaces (slices). 
This is an obstacle for many results and constructions with these types of spacetimes, and is particularly worrisome since it is not known whether spacetimes that do have CMC slices are in any sense generic.
In this expository paper, we will discuss the known results about the existence (and non-existence) of CMC slices, examine the evidence for cases which are unknown, and make several conjectures concerning the existence of CMC slices and their generality.
\end{abstract}

\vspace*{-1.5cm}
\maketitle
\tableofcontents
\vspace*{-1.5cm}

\section{Introduction}
\label{Intro}

As is well known, the problem of finding solutions to the Einstein Equations of general relativity can be split into finding initial data satisfying the constraint equations, and then evolving that data using the evolution equations.
Both portions have their own interesting problems.
What kinds of initial data are or are not possible?
Is there a way to parameterize all initial data? Is the evolution problem stable? 
What happens near singularities, and what kinds of singularities can occur?
Much progress has been made in recent decades on these kinds of problems. But looking at many of these papers, a common theme appears:
\emph{Constant mean curvature (CMC) initial data makes everything easier.}

In finding initial data in general relativity, the commonly used conformal method involves solving a system of coupled equations (the conformal constraint equations), in which the mean curvature function is one of the freely-specified parameters. 
If the chosen mean curvature is constant, then the equations decouple, making finding solutions much easier. 
Because of this, the CMC case of the initial data question is mostly understood.
For instance, a parameterization of all closed CMC initial data has been found (see \cite{Isenberg95}).
However, the conformal method in the non-CMC case is a completely different story, and is one of the main motivations for this study; we will come back to this shortly.

In the evolution problem in general relativity, there is some coordinate freedom; you may pick a lapse and shift arbitrarily. 
Obviously, some choices are better adapted to proving estimates than others.
If you start with CMC initial data, there is a choice of lapse and shift that allow your time function to be the (constant) mean curvature of each time slice.
This choice ends up 
being useful in proving various key estimates. 
For example, it is used in the recent resolution of the bounded $L^2$ curvature conjecture by Klainerman, Rodnianski, and Szeftel \cite{KRS16}, which states, essentially, that the evolution of the spacetime can continue as long as the $L^2$ norm of the curvature remains bounded.

In both of these examples of results for the Einstein constraint and evolution equations, the authors make the assumption that the spacetimes they are considering have a CMC slice (Cauchy data). 
Unfortunately, not all spacetimes have CMC slices!
This was first observed by Bartnik in \cite{Bartnik88}, and vacuum examples were later found by Chru\'sciel, Isenberg, and Pollack in \cite{CIP05}.
This has serious implications for many of the results that have been obtained for both the Einstein constraint and evolution equations.

For those investigating initial data through study of the Einstein constraint equations, one major goal is the parameterization of all initial data.
The conformal method works well as this parameterization for CMC initial data but, unfortunately, it fails for non-CMC data.
The solution space of the conformal constraint equations for far-from-CMC data exhibit folds and blowup in unexpected places. 
This has been shown analytically in some cases \cite{Maxwell11, Nguyen15}, and numerically in others \cite{DHKM17}.
The behavior of the conformal constraint equations in the far-from-CMC regime is very complicated. 
However, if every spacetime had a CMC Cauchy surface (slice), while the conformal method would not parameterize all initial data, it would at least parameterize initial data for every spacetime.
Unfortunately, not all spacetimes have CMC slices, but a similar statement could be made as long as spacetimes with CMC slices are generic.

Many of the most important conjectures in mathematical relativity ask whether generic spacetimes have (or don't have) certain properties or features, such as Cauchy horizons.
However, in many papers about the evolution problem, assuming the existence of a CMC slice is a key first step in the analysis.
While this may be more a limitation of methods rather than an actual obstruction, it is still worrying.
These results (such as the $L^2$ curvature theorem) cannot be used to prove anything about generic spacetimes, unless it is known that a generic spacetime has a CMC slice.
That is currently unknown.
Additionally, is it possible that the spacetimes without CMC slices are an obstruction that makes proving these results more difficult?
The main concern is that it is unknown how general spacetimes without CMC slices are.
While there are a few examples, they are very special examples, with high levels of symmetry.
Is there an open set of such spacetimes? Are they ubiquitous, or very special?

In this expository paper, we will discuss known results about the existence and non-existence of constant mean curvature slices in spacetimes, possible directions for investigation, and make a few conjectures.

\section{Technical Background}
\label{techIntro}

A Lorentzian manifold $\calV$ is a \emph{cosmological spacetime} if it is
globally hyperbolic with compact Cauchy surfaces and satisfies the 
\emph{timelike convergence condition}, 
\begin{equation}\label{def:TCC}
  R_{\mu\nu} T^\mu T^\nu \geq 0\textrm{, for every timelike vector } T.
\end{equation} If the spacetime also obeys the Einstein equations, this is
equivalent to the strong energy condition. Many results included will not
require that $\calV$ satisfies the Einstein equations, though we are 
predominately interested in their solutions. By \emph{slice}, we will
always mean a compact Cauchy surface. 

Intuitively, the mean curvature gives the expansion and contraction rate
of the universe. We will use the sign convention that $\tr k = H<0$ means the universe is 
contracting to the future. It is well known \cite{HE} that if $H<0$ on a slice,
then all future timelike geodesics must be incomplete (in the globally hyperbolic
development), with a uniform upper bound on their proper length of $3/-H$.

There are a number of well known, useful facts about CMC slices 
in cosmological spacetimes, when they
exist (see \cite{Bartnik88, MT80}). For instance,
if a CMC slice exists in $\calV$, with $H \neq 0$, it is necessarily unique.
For maximal slices, with $H \equiv 0$, we can only say that the spacetime is static,
with a timelike Killing field $T$ with $R_{\mu\nu}T^\mu T^\nu = 0$. 

When a spacetime has a CMC slice, there is a foliation, at least locally, 
by CMC slices. The mean curvatures vary along the foliation monotonically, except,
sometimes, when $H\equiv 0$. In the maximal case, nearby slices may also be maximal,
but if they are, the spacetime is static. A major question (see, for instance, 
\cite{Rendall97}) is whether this foliation covers the spacetime, and whether they
necessarily achieve all possible mean curvatures.\footnote{If the topology of the 
slices do not allow a metric with positive scalar curvature, the spacetime
cannot have a maximal slice. Thus, all allowable mean curvatures would be $(-\infty,0)$
or $(0,\infty)$.}

For any two slices with mean curvatures $H_1$ and $H_2$, and one in the future 
of the other, then for any mean curvature $H$ satisfying 
$\sup H_1 \leq H \leq \inf H_2$, then there is a slice with mean curvature $H$
in the spacetime, between the other slices.

One common use for this result is that, if $\calV$ has CMC slices of mean
curvature $H_1\leq H_2$, then the region between the slices is 
foliated by CMC slices, with mean curvatures monotonically increasing along
the foliation from $H_1$ to $H_2$.

It will also be useful to define a spacetime ray and line.

\begin{defn}
  A future (resp. past) \emph{ray} is a future (resp. past) timelike path
   with infinite affine length.
\end{defn}
\begin{defn}
  A \emph{line} is a a timelike \emph{geodesic} with infinite affine length
   to the future \emph{and} past, which is also \emph{globally} maximizing in distance.
\end{defn}

The earlier mentioned result, that if there is a slice with $H<0$, then timelike
geodesics have an upper bound on length to the future, shows that the existence
of a CMC implies the nonexistence of rays, either to the past or future, depending 
on the sign. On the other hand, the importance of lines is shown by the following
theorem from \cite{Bartnik88}.

\begin{thm}\label{thm:LineMeansProduct}
  A globally hyperbolic spacetime satisfying the timelike convergence criterion
  \eqref{def:TCC} and admitting a line is a metric product. Thus, it has maximal 
  slices.
\end{thm}

\section{Which Spacetimes have CMC Slices?}

The most basic question to ask is, ``Under what conditions does a spacetime
have or not have CMC slices?"

Since CMC slices are unique, it is easy to cut enough out of a spacetime to create 
a spacetime without a CMC slice. We are not concerned with such examples. In order to
avoid that, we will consider only maximal globally hyperbolic developments (MGHDs).

Unfortunately, since work on barrier methods in the 80's, by Bartnik and others,
there are relatively few papers
that deal with the existence of a CMC slice. (Most papers assume three is one, then
investigate the foliation by CMC slices.) Various barrier methods are well developed.
The barriers are essentially used to guarantee that the slices avoid the singularities 
at future or past infinity.

One type of barrier was used in \cite{Bartnik88}. In that paper, he showed
that if $\calV - I(p)$ is compact for some point $p\in \calV$, then there is a
CMC slice, passing through $p$. The barrier in the proof is $\partial I(p)$,
which is also compact. He uses this boundary to show that there are CMC slices
through $p$, though singular at $p$, with any desired mean curvature. He then
shows that one of these slices must in fact be smooth.

The condition that $\calV-I(p)$ is compact is an interesting condition.
It roughly states that an observer at $p$ could observe all events that 
happened in the universe, sufficiently far in the past. There is no 
``hidden" portion of the universe. A similar statement can be said about the
future. However, it is not a necessary condition. Spherical FRLW spacetimes 
don't satisfy it (light rays emitted from the big bang don't travel across the 
universe before recollapse, and so there can be no such $p$),
but they still have CMC slices.

A more common kind of barrier is a slice with a certain mean curvature.
As discussed in Section \ref{techIntro}, if we can find two slices,
one in the future of the other, and with
mean curvatures $\sup H_1 \leq \inf H_2$, then there is a CMC slice between them.

This immediately shows that a spacetime with a crushing singularity has CMC slices.
A (future) \emph{crushing singularity} means that there is 
a sequence of slices $\Sigma_i$
approaching the future edge of the MGHD, with mean curvatures $H_i$ 
satisfying $\sup_{\Sigma_i} H_i \to -\infty$. If you have a crushing singularity,
pick any slice $\Sigma$ as one barrier, then use the crushing singularity
to find a slice $\Sigma'$ with $\sup_{\Sigma'} H'\leq \inf_\Sigma H$.
A CMC slice is then between them.

It is then useful to find other conditions which guarantee that your future
singularity is crushing. One such condition is the strong curvature 
singularity.

\begin{defn}
  A timelike geodesic $\lambda(t)$ terminates in a \emph{strong curvature singularity} at
  affine parameter value $t_0$ if the following holds: Let $\mu(t)$ be the three-from
  on the normal space to $\lambda'(t)$ determined by any three linearly independent
  vorticity-free Jacobi fields along $\lambda(t)$. If $\mu(t)$
  vanishes for at most finitely many $t$ is some neighborhood $[t_1, t_0)$ of $t_0$,
  then we require $\lim_{t\to t_0} \mu(t) = 0$. 
\end{defn}

For cosmological spacetimes, if $\mu(t)$ vanishes finitely many times, the limit exists.
The basic equation is 
\begin{equation}
  \frac{d^2}{dt^2} \mu(t)^{1/3} + \frac13(R_{ab} v^a v^b + 2 \sigma^2)\mu^{1/3} = 0,
\end{equation} and, importantly, note that the second term is positive, which
makes $\mu$ behave well. Using this equation, it can be shown that any global strong
curvature singularity is also a crushing singularity.

\begin{thm}\cite[Thm 6]{MT80} 
  For a cosmological spacetime, if there exists a Cauchy surface from which all 
  orthogonal (future or past) 
  timelike geodesics end in a strong curvature singularity, then
  that singularity is crushing, and so there is a CMC slice.
\end{thm}

Another theorem in that same paper gives another condition guaranteeing that a
singularity is crushing. It has to do with the $c$-boundary. The future 
$c$-boundary is essentially constructed by associating each point $p$ of the spacetime
with its past $I^-(p)$. The future boundary is then defined as the sets $I^-(\gamma)$,
where $\gamma$ is an inextendible future directed timelike path. (The technical details
for these ``indecomposable pasts" can be complicated, but we won't need them in
this paper.)

\begin{thm}\cite[Thm 6]{MT80}
 Suppose a cosmological spacetime with global future singularity has $c$-boundary
 consisting of a single ``point," i.e., that all inextendible future timelike paths
 have the same past (and thus, that past is all of the spacetime.) 
 Then the singularity is crushing, and so a CMC slice exists.
\end{thm}

Those are the most general results that we are aware of for proving that a CMC
slice exists.

Let's now turn our attention to the known examples of spacetimes without 
CMC slices. Unfortunately, only a few examples are known.

Though this paper is predominantly focused on spatially compact spacetimes, we should
briefly mention results for spatially asymptotically Euclidean manifolds. In these
spacetimes, an CMC asymptotically Euclidean Cauchy surface can only have $H\equiv0$,
i.e., be maximal. If the spacetime obeys the weak energy condition, then the
Cauchy surface must have a metric of positive scalar.

However, just as in the compact case, there are topological 
obstructions to this. An explicit example was found by Brill in \cite{Brill82}, but it
is now known \cite{DM15}
that any asymptotically Euclidean manifold, if it can be one point
compactified (one point for each end) to a manifold not allowing a metric of positive
scalar curvature, itself cannot have an asymptotically Euclidean metric of positive 
scalar curvature. This gives well-understood topological restrictions to the existence
of maximal Cauchy surfaces in spatially asymptotically Euclidean spacetimes.

For spatially compact spacetimes, the original, explicit example
comes from Bartnik's paper \cite{Bartnik88}. To construct his example, take the
maximally extended Schwarzschild spacetime. Then, on each of the ends, cut off each
end and attach a torus with a homogeneous dust in it. Bartnik uses a particular model
(Tolman-Bondi)
that allows gluing between these regions in such a way that the spacetime
evolves nicely and satisfies the strong energy condition. 

Bartnik proves that this spacetime has no CMC slices in two ways. 

The first way is by symmetry. The spacetime has a ``time inversion" symmetry. 
In particular, if there is a non-maximal CMC slice, then by this symmetry 
there is one with the same mean curvature but opposite sign.

In between these slices, you can foliate the spacetime with CMC slices, and thus
find a maximal slice. But the spacetime satisfies the weak energy condition, and so 
the maximal slice has nonnegative scalar curvature by the constraint equations.
This contradicts the topology of the slices, $T^3\# T^3$.

This proof method is simple, and shows that the topology of the slices may be
important in proving more general results.
However, symmetry is obviously not a generic property, and so this proof seems
limited to very symmetric data, a serious limitation.

However, it also allowed the development of a family of examples in \cite{CIP05}.
In this paper, they take $T^3$ vacuum solutions,
then glue them (using IMP gluing) symmetrically. The important addition is that 
these examples are vacuum, and so the dust from Bartnik's example was not necessary.
\footnote{It is known that dust can cause shell-crossing singularities (though there
are none in Bartnik's example), and so dust is
sometimes avoided.}
However, it is still hard to generalize these examples due to the symmetry requirement.

Bartnik's second proof uses timelike path incompleteness.

\begin{thm}[\cite{Bartnik88}]\label{thm:rays?NoCMC}
  If a globally hyperbolic, cosmological
  spacetime $\calV$ has a future ray and a past ray, but no line,
   then $\calV$ has no
  CMC Cauchy surfaces.
\end{thm}
\begin{proof}
Suppose $\calV$ did have a CMC slice, 
  with $H\not\equiv 0$. Then by \cite[pg 274]{HE}, all future (or past) timelike
  paths have finite length (depending on the sign of $H$.) This is a contradiction. 
  
  If $H\equiv 0$, the spacetime is locally foliated either by CMC slices
  with nonzero $H$ (in which case the first paragraph applies),
  or it is foliated by maximal slices. If this maximal
  foliation covers the spacetime, the spacetime is static and thus 
  contains a line. This is a contradiction.
\end{proof}

The explicit example constructed by Bartnik has rays (complete timelike paths
to the past or future), because the Schwarzschild spacetime does. However, it clearly
has no line (and, in fact, no complete timelike geodesics at all), and so cannot
have a CMC slice.

As this second proof doesn't use symmetry, 
this approach seems a more likely candidate
for generalization. Indeed, our main conjecture is that this condition is
both necessary and sufficient. 

\begin{conj}
 A globally hyperbolic, cosmological spacetime has no CMC Cauchy surfaces 
 if and only if it 
 contains a future and a past ray, but no line.
\end{conj}

To see why this conjecture is reasonable, recall Theorem \ref{thm:LineMeansProduct},
which says that if a spacetime has a line, then your spacetime is a metric product, 
and thus have maximal slices.
The difficult part of the conjecture is then to show the following:

\begin{conj}\label{mainConj}
  If, in a cosmological spacetime, every future (or past) timelike path is incomplete
  (i.e., of finite length), then there exists a CMC slice.
\end{conj}

In this formulation the conjecture essentially says that,
 if you have a global singularity, then you
have a CMC slice. 

For certain kinds of singularities, such as the crushing and strong
curvature singularities we discussed earlier, we know this conjecture holds. 
Unfortunately, not all singularities are necessarily of the crushing type. For 
instance, in \cite[pg 3594]{Rendall97}, Rendall constructs an example
where a dust shell-crossing singularity causes a CMC foliation to stop at a finite
mean curvature, and so the
singularity is not crushing. Of course, this example uses dust, which is known
to cause problems,
and still has CMC slices.  We are aware of no other cosmological MGHDs,
ending in a global
singularity, which do not have a crushing singularity.

Both of Bartnik's proofs in fact prove something stronger than that the spacetime
has no CMC slices. They also show that the spacetimes
have no slices of constant signed mean curvature. This observation leads to
a conjecture.

\begin{conj}\label{conj:CMC=constSign}
  A cosmological spacetime has a CMC slice if and only if it has a slice of constant
  signed mean curvature.
\end{conj}

Since a slice of constant (non-zero) signed mean curvature implies that there can 
be no rays, the main conjecture \ref{mainConj} would imply this conjecture. However,
it provides an interesting intuitive picture for spacetimes without CMC slices.

If this conjecture is true, then every slice must have a region
with positive mean curvature and
a region with negative mean curvature. Because of this, 
there appears to be one region of the spacetime that expands for
all time, and another that expands to the past for all time. (This also lines up
with the existence of rays, as in the main conjecture \ref{mainConj}.)
If these regions were
causally linked, you might expect to find a line going from the one region to the
other, which would lead to a contradiction. Thus the spacetime has (at least) two regions
separated by some sort of ``wormhole," as in Bartnik's example, with on region expanding,
the other contracting. 

The one implication of the conjecture \ref{conj:CMC=constSign} is obvious.
 he other makes sense heuristically. If $\calV$ has
a slice of constant signed (non-maximal) mean curvature, there is a global singularity
to the future or past. By the previous conjecture, there must be a CMC slice.

More directly, suppose that 
$\calV$ has a slice of constant signed (non-maximal) mean curvature. 
It is possible to evolve the slice
by mean curvature flow. On a Lorentzian manifold, the mean curvature flow tends to
expand slices, and so the slice will evolve into the expanding universe direction, i.e.,
towards $H\equiv 0$. As the universe is expanding, you would not expect
singularities to appear before you reach the point where the new mean curvature
is globally closer to zero than the original mean curvature.\footnote{Indeed, this 
can almost be carried out. Using basic estimates, it is not hard to show that 
as long as all timelike geodesics exist for a certain explicit amount of proper time
into the expanding direction (based \emph{only} on the starting mean curvature),
then you can evolve the slice long enough to get the desired slice with small mean 
curvature. Unfortunately, in trying to show that you can always evolve long
enough to find such a slice, you need an estimate of certain quantities. The type
of estimate that would seem to work is exactly the kind of estimate that CMC slices
are used for in evolution problem papers--a relationship between the second fundamental
form and the lapse. Thus the problem becomes somewhat circular.}
Using the standard interpolation result, there is a slice of constant mean curvature
between them.

The most useful condition for showing the existence or nonexistence of a CMC slice
would refer only to the initial data. 
For instance, if a set of vacuum initial data is of 
Yamabe class $Y^+$, the closed universe recollapse conjecture \cite{BGT86}
says that the spacetime 
will begin and end in a global singularity. If this conjecture is true, and
the main conjecture \ref{mainConj} is true, then every Yamabe positive vacuum
initial data set
leads to a spacetime with CMC slices. 

Another idea is to try to find conditions on a region of initial data such that 
that region, when evolved, will contain rays. For instance, like in Bartnik's example,
if there is a region of the initial data with positive mean curvature, 
and is is separated from other regions by a black-hole-type region, then it seems
reasonable that this region should expand for all time, and thus have rays. Then,
if an initial data set had one region expanding to the future in this way, and one
to the past, then, by Theorem \ref{thm:rays?NoCMC}, the spacetime cannot have any CMC
slices. This condition would be especially interesting if it were stable under
perturbations, which it seems like it should be.

The difficulty in any attempt to find conditions on initial data is that 
proving that there are or are not CMC slices in the spacetime requires asking 
questions about the long time behavior of the spacetime. Results of this nature
tend to be difficult to prove, and many of the best results require the existence
of a CMC slice. This difficulty is the reason that the only known condition
on initial data is the very restrictive time-antisymmetric condition used earlier.

\section{Are CMC Slices Generic?}

The most important question about CMC slices is their genericity. If the existence
of CMC slices is generic, then the special examples where they do not exist can
mostly be ignored, and CMC slicing can be used for proving generic properties.
Unfortunately, to our knowledge, 
there no evidence that this is the case, other than the weak evidence that all
known examples are all very special, all on $T^3\# T^3$ with time-antisymmetry.

In any question of genericity, it is vital to choose a useful and meaningful topology.
In questions about spacetimes, there are several choices that could be made.

The simplest choice would be a global $C^{2,\alpha}$ topology. Using this
topology, Choquet-Bruhat \cite{CB76}
 showed that spacetimes with CMC slices form an open set.
In particular, she proves that perturbations in a neighborhood of the evolution of a
CMC slice always has another CMC slice. This is done using the linearization of the
mean curvature operator.

However, in this topology, the existence of CMC slices is not generic!
Due to Theorem \ref{thm:rays?NoCMC}, if there are timelike rays, but no timelike line,
then the spacetime does not have a CMC slice. However, the unboundedness of rays is
stable under a global $C^{2,\alpha}$ perturbation. Also, if there were a line, 
Theorem \ref{thm:LineMeansProduct} says the spacetime would be static.
However, Bartnik's example is not $C^{2,\alpha}$ close to a
static spacetime, and so  a small enough $C^{2,\alpha}$ perturbation could
not create a line.
Together, this means that every small, global $C^{2,\alpha}$ perturbation of Bartnik's 
example would still have no CMC slices.

We would be among the first to admit, though, that the global $C^{2,\alpha}$ topology,
along with any other global topology, is probably not the right one for this problem.
For instance, since we are dealing only with MGHDs, spacetimes end in singularities or
exist for all time. In either case, something is unbounded.
Thus, $C^{2,\alpha}$ close is extremely restrictive.

A related problem is that there are one parameter families of initial data which
are not close to each other in the $C^{2,\alpha}$ norm. For example,
initial data for the flat toroidal, static universe can be perturbed into initial
data which evolve into spacetimes with either global future or global past singularities.

It seems that a topology based on initial data is more reasonable for this problem. 
We propose to use the $C^{2,\alpha}$ topology on initial data sets $(M, g, K)$. 

Using this topology, the existence of CMC slices is still an open condition.
Consider initial
data evolving into a spacetime with at least one CMC slice. That CMC slice is within
some amount of coordinate time $t_0$ from the initial slice. 
For any initial data sufficiently close to the original data,
the spacetimes are $C^{2,\alpha}$ close
within $t_0$ of the initial data. Again using Choquet-Bruhat's argument in \cite{CB76},
the nearby spacetime also has a CMC slice, of the same mean curvature.

Unfortunately, as with the global $C^{2,\alpha}$ topology, it appears that 
the existence of CMC slices is not a generic condition.

\begin{conj}
There is an open set of initial data (in the $C^{2,\alpha}$ norm) such that the
associated spacetimes do not have CMC slices.
\end{conj}

Consider Bartnik's example that we discussed earlier. To
remind the reader, his spacetime is two toroidal universes, (one expanding to the
future, the other to the past,) separated by a Schwarzschild bridge. If the existence
of CMC slices was a dense condition, a generic perturbation of initial data for
Bartnik's example should evolve into a spacetime with CMC slices.

This perturbed data could not lead to a static spacetime, and so it has a non-maximal 
CMC slice. Thus, all future (or past) geodesics are incomplete. In other words, a
generic perturbation causes at least one of the expanding toroidal universes to
collapse. But a small perturbation shouldn't cause global collapse. To make this
a bit more precise, we make the following conjecture.

\begin{conj}\label{conj:PerturbationsDon'tCauseCollapse}
If a compact, vacuum initial data set evolves into a non-static spacetime such that there is a 
family of future rays intersecting the initial Cauchy surface in an open set, then any 
small $C^{2,\alpha}$ perturbation of the initial data evolves into a spacetime
with at least one future ray.
\end{conj}

This is not the same as stability of the spacetime. For instance, the perturbation
could coalesce into a black hole type region, but as long as it didn't cause the
global collapse of the expanding region, the conjecture would still hold. 

If this conjecture were true for initial data with dust, Bartnik's example would 
immediately imply that the existence of CMC slices is not generic. Similarly,
the examples of Chru{\'s}ciel, Isenberg, and Pollack, if they contain a ray, would
also imply that the existence of CMC slices is not generic.

Another heuristic example can be found by gluing together two stable 2-tori. 
In two papers, \cite{AM04, AM11}, Andersson and Moncrief give vacuum, hyperbolic FLRW models with
the metric $-d\rho^2 + \rho^2 \gamma$, where
$\gamma$ is the hyperbolic metric on the 2-torus slice. 
Importantly, they prove these spacetimes
are stable attractors of nearby initial data, at least in the expanding direction.
If we take two of the same constant-$\rho$ slices, but change the sign of
the second fundamental form $k$ on the
second one, we can do the same gluing they used in \cite{CIP05} in order to glue
these together anti-symmetrically in $k$ and symmetrically in the metric. Then, as with
their examples, the evolution can have no CMC slices due to the symmetry and bad topology.

Of course, these examples are really just a special case of the examples in \cite{CIP05}.
The advantage these have is that the base hyperbolic space we are using is known to
be stable. That means, after the gluing, in the domain of dependence 
(in the expanding direction) away from the
gluing region, the spacetime is close to, and in fact converging toward, the original
FLRW. If this domain existed for all time under any perturbation,
the new spacetime would have infinite
length timelike paths, and thus no CMC slices. That would show that the set of initial
data that evolve to have CMC slices is \emph{not} dense!

Unfortunately, it is unclear whether this domain of dependence lasts for all time.
In the original FLRW
spacetime, it is straightforward to check that the expansion rate of the universe
is just barely too 
slow to guarantee that. (If the $\rho^2$ were replaced by $\rho^2 \ln^3(\rho)$,
for instance, it would work.)

However, perturbations of this glued spacetime likely still have no CMC slices.
Heuristically, the wormhole bridge connecting 
the two glued tori should act like the Schwarzschild bridge in Bartnik's example.
Since it is a black-hole-like region, we would expect it to stay small, only lightly
affecting the spacetime outside its immediate neighborhood.
If this were true, the spacetime would not recollapse in the expanding 
direction, and so all perturbations would lead to spacetimes without CMC slices.

It is, of course, possible that spacetimes with CMC slices are generic.
If that were true, for any initial data set $S$ leading to a spacetime without 
CMC slices, there are sets of initial data $S_i$,
leading to spacetimes \emph{with} CMC slices, converging to the initial data in the 
$C^{2,\alpha}$ topology. Since the corresponding spacetimes $\calV_i$ 
will then converge on 
any compact interval of time (of the spacetime $\calV$ evolved from $S$), 
if the CMC slices of the $\calV_i$ were well behaved, then the boundary spacetime 
$\calV$ would also contain a CMC slice. Thus, the CMC slices of the $\calV_i$ cannot
be well behaved.

There are three ways the convergence of CMC slices could fail. 
The first is that for each $H$,
the CMC slice in $\calV_i$ of mean curvature $H$ could ``run off" to infinity, so
that, in the limit, no compact interval of time could contain the slices. The second
is that these slices may become null in the limit. For the third,
it is possible that, for any $H$,
for $i$ large enough, $\calV_i$ does not have a CMC slice of mean curvature $H$.
In other words, the foliation by CMC slices in $\calV_i$ covers a vanishing
interval of the possible mean curvatures.

In any of these cases, the boundary spacetime $\calV$ should have slices
of arbitrarily near constant mean curvature. If the converging slices are
non-maximal, the slices in $\calV$ should be near-CMC in the traditional sense
of $\|\nabla H\|$ being small compared to $\|H\|$. If the converging slices are
maximal (or approach maximal), as they would be in Conjecture \ref{conj:CMC=constSign}
were true, the slices in $\calV$ would be near-CMC but perhaps only in the sense
that the average of $H$ could be made arbitrarily small.

Thus, heuristically, the existence of CMC slices is a generic property if and only
if every spacetime has near-CMC slices. Thus, one way to check whether we should
expect this property to be generic is to check whether the known examples have 
arbitrarily near-CMC slices. If they do not, we would expect that the existence of
CMC slices is not generic. If they do, we cannot, however, conclude the opposite.

Let us mention that this heuristic argument has consequences for the initial data 
problem. The conformal method, as mentioned in the introduction, parameterizes
all CMC initial data. While it fails for far-from-CMC data, it still behaves well for
near-CMC data, including existence and uniqueness. 
If it were true that all spacetimes have near-CMC slices, the conformal method
may still be able to provide a reasonable 
parameterization of initial data for all spacetimes.

While genericity is the main question, it is not even known whether spacetimes with
CMC slices and those without are even in the same connected component. To prove that
they aren't, one could try to show that the no-CMC-slice condition is open, or 
that the CMC-slice condition is closed. Both of those methods, if our main Conjecture
\ref{mainConj} is true, are questions about long term stability of the spacetime:
If a spacetime has a ray, do all nearby spacetimes also have a ray? 

At first glance it may seem like a spacetime with a single (future) ray could be
perturbed to a spacetime without a ray. However, do such spacetimes even exist?
By conjecture \ref{conj:CMC=constSign}, which says that spacetimes without
CMC slices will have no slices of constant sign, we expect that a spacetime
without CMC slices will be expanding in some region. An expanding region should
have an open set of rays.

Additionally, an expanding
region should be stable. The boundary example of the static universe is unstable, 
and cannot occur since there is a region that is expanding to the past as well.
Any sufficiently small perturbation of the expanding region, as we've argued before,
should not cause the expanding region to collapse. This would suggest that having
no CMC slices is also an open condition, and thus forms a disconnected
component of initial data.

One could also try to show connectedness direction.
The most obvious way is to take the initial data for an example with no CMC slices, and treat it as seed data for the conformal constraint equations. Then one can
make a path of seed data from that given data to one with constant mean curvature. 
The difficulty is that one must solve the conformal constraint equations for each set
of seed data in order to find initial data for a spacetime. Unfortunately, it is
now well established (see \cite{DHKM17}) that for far-from-CMC seed data, which we,
of necessity, would have to work with, the conformal constraint equations are
very complicated, and it is unclear that they have solutions for that seed data.

\section{Final Comments}

Constant mean curvature Cauchy surfaces are useful and much easier to work with than
their non-CMC counterparts. This is true both for initial data and for the evolution
of the data. Unfortunately, while much is known about foliations by CMC slices,
\emph{given} a starting CMC slice, less has been written about whether or not
spacetimes have CMC slices at all.

Importantly, many results assume that the spacetime has a CMC slice. Unfortunately,
it is currently unknown whether or not that assumption is true for a generic spacetime.
Our conjecture is that it is \emph{not} a generic condition. Since many of the most
important conjectures in mathematical relativity concern generic spacetimes, this
means that results that assume the existence of a CMC slice are not applicable to
those problems, unless
the assumption of the existence of a CMC slice can somehow be removed.



\bibliographystyle{abbrv}
\bibliography{MeanCurvature}

\begin{thebibliography}{10}

\bibitem{AM04}
L.~Andersson and V.~Moncrief.
\newblock Future complete vacuum spacetimes.
\newblock In {\em The {E}instein equations and the large scale behavior of
  gravitational fields}, pages 299--330. Birkh\"auser, Basel, 2004.

\bibitem{AM11}
L.~Andersson and V.~Moncrief.
\newblock Einstein spaces as attractors for the {E}instein flow.
\newblock {\em J. Differential Geom.}, 89(1):1--47, 2011.

\bibitem{BGT86}
J.~Barrow, G.~Galloway, and F.~Tipler.
\newblock The closed-universe recollapse conjecture.
\newblock {\em Mon. Not. R. assr. Soc.}, 223:835--844, Dec 1986.

\bibitem{Bartnik88}
R.~Bartnik.
\newblock Remarks on cosmological spacetimes and constant mean curvature
  surfaces.
\newblock {\em Comm. Math. Phys.}, 117(4):615--624, 1988.

\bibitem{Brill82}
D.~Brill.
\newblock On spacetimes without maximal surfaces.
\newblock {\em In: Proc. Third Marcel Grossman meeting}, 1982.
\newblock Ed: Ning, H. Amsterdam: North-Holland.

\bibitem{CB76}
Y.~Choquet-Bruhat.
\newblock Maximal submanifolds and submanifolds with constant mean extrinsic
  curvature of a {L}orentzian manifold.
\newblock {\em Ann. Scuola Norm. Sup. Pisa Cl. Sci. (4)}, 3(3):361--376, 1976.

\bibitem{CIP05}
P.~T. Chru\'sciel, J.~Isenberg, and D.~Pollack.
\newblock Initial data engineering.
\newblock {\em Comm. Math. Phys.}, 257(1):29--42, 2005.

\bibitem{DHKM17}
J.~Dilts, M.~Holst, T.~Kozareva, and D.~Maxwell.
\newblock Numerical bifurcation analysis of the conformal method.
\newblock Preprint.

\bibitem{DM15}
J.~Dilts and D.~Maxwell.
\newblock {Y}amabe classification and prescribed scalar curvature in the
  asymptotically {E}uclidean setting.
\newblock 2015.
\newblock arXiv:1503:04172, accepted by Communications in Analysis and
  Geometry.

\bibitem{HE}
S.~W. Hawking and G.~F.~R. Ellis.
\newblock {\em The large scale structure of space-time}.
\newblock Cambridge University Press, London-New York, 1973.
\newblock Cambridge Monographs on Mathematical Physics, No. 1.

\bibitem{Isenberg95}
J.~Isenberg.
\newblock Constant mean curvature solutions of the {E}instein constraint
  equations on closed manifolds.
\newblock {\em Classical Quantum Gravity}, 12(9):2249--2274, 1995.

\bibitem{KRS16}
S.~Klainerman, I.~Rodnianski, and J.~Szeftel.
\newblock The resolution of the bounded {$L^2$} curvature conjecture in general
  relativity.
\newblock {\em Bull. Braz. Math. Soc. (N.S.)}, 47(2):445--456, 2016.

\bibitem{MT80}
J.~E. Marsden and F.~J. Tipler.
\newblock Maximal hypersurfaces and foliations of constant mean curvature in
  general relativity.
\newblock {\em Phys. Rep.}, 66(3):109--139, 1980.

\bibitem{Maxwell11}
D.~Maxwell.
\newblock A model problem for conformal parameterizations of the {E}instein
  constraint equations.
\newblock {\em Comm. Math. Phys.}, 302(3):697--736, 2011.

\bibitem{Nguyen15}
T.-C. Nguyen.
\newblock Nonexistence and nonuniqueness results for solutions to the vacuum
  {E}instein conformal constraint equations.
\newblock \sf arXiv:1507.01081 [math.AP].

\bibitem{Rendall97}
A.~D. Rendall.
\newblock Existence and non-existence results for global constant mean
  curvature foliations.
\newblock In {\em Proceedings of the {S}econd {W}orld {C}ongress of {N}onlinear
  {A}nalysts, {P}art 6 ({A}thens, 1996)}, volume~30, pages 3589--3598, 1997.

\end{thebibliography}
\end{document}